%% file: main-anpp.tex
\documentclass[11pt]{article}
\input{macros}

\title{Optimal Auctions vs.\@ Anonymous Pricing}
\author{Saeed Alaei, Jason Hartline, Rad Niazadeh, Emmanouil Pountourakis, and Yang Yuan}

\date{}

\begin{document}
\maketitle


\input{abstract}

\thispagestyle{empty}
\pagenumbering{arabic} 
\input{intro}

\input{prelim}

\input{upperbound}

\input{lowerbound}
\input{irregular}

\input{simulation}
\section*{Acknowledgment}
The authors would like to express their very great appreciations to Prof. Robert D. Kleinberg for his valuable and constructive suggestions.

\bibliographystyle{apalike}
\bibliography{auctions}
\appendix
\input{upperbound-appendix}

\end{document}

%% file: macros.tex
\usepackage{amsmath,amsfonts,amssymb,amsthm}
\usepackage{mathrsfs}
\usepackage{dsfont}
\usepackage{ifthen}
\usepackage{commath}
\usepackage{graphicx}
\usepackage{enumerate}
\usepackage{color}
\usepackage[margin=1in]{geometry}
\usepackage{natbib}
\usepackage{comment}

\usepackage{tikz}
\usepackage[pdfpagelabels]{hyperref}
\usepackage[capitalize]{cleveref}

\usepackage{pgfplotstable}
\usepackage{pgfplots}

\usetikzlibrary{arrows,decorations,backgrounds}

\newtheorem{theorem}{Theorem}[section]
\newtheorem{lemma}[theorem]{Lemma}
\newtheorem*{lemma*}{Lemma}
\newtheorem{proposition}[theorem]{Proposition}

\theoremstyle{definition}
\newtheorem{definition}{Definition}[section]
\DeclareMathAlphabet{\mathbfit}{OML}{cmm}{b}{it}

\newcommand{\suchthat}{\;\ifnum\currentgrouptype=16 \middle\fi|\;}

\newcommand{\myref}[2]{\hyperref[#2]{$#1$\ref*{#2}}}

\newcommand{\AutoAdjust}[3]{\mathchoice{ \left #1 #2  \right #3}{#1 #2 #3}{#1 #2 #3}{#1 #2 #3} }

\newcommand{\RangeMN}[2]{\AutoAdjust{\{}{{#1},\ldots,{#2}}{\}}}
\newcommand{\RangeAB}[2]{\AutoAdjust{[}{{#1},{#2}}{]}}
\newcommand{\RangeAb}[2]{\AutoAdjust{[}{{#1},{#2}}{)}}

\newcommand{\Rangeab}[2]{\AutoAdjust{(}{{#1},{#2}}{)}}
\newcommand{\RangeN}[1]{\InBrackets{#1}}

\newcommand{\Reals}{\mathbb{R}}

\newcommand{\PosReals}{\Reals_+}

\newcommand{\PosInts}{\mathbb{N}}

\usepackage{natbib}

\newcommand{\NumAgents}{n}

\newcommand{\Rev}{R}

\newcommand{\MaxVal}{{\bar{v}}}
\newcommand{\MaxVals}{{\bar{\mathbf v}}}
\newcommand{\MaxValD}{u}
\newcommand{\MaxValC}{v}

\newcommand{\MaxProb}{{\bar{q}}}
\newcommand{\MaxProbs}{{\bar{\mathbf q}}}
\newcommand{\Prob}{q}
\newcommand{\Price}{p}

\newcommand{\PPrice}{p}
\newcommand{\SPPrice}{{p^*}}

\newcommand{\TriDist}{\operatorname{Tri}}
\newcommand{\Val}{v}

\newcommand{\BenchRatio}{\operatorname{\textsc{Bound}}}
\newcommand{\BenchRatioI}{\BenchRatio_{I}}
\newcommand{\BenchRatioII}{\BenchRatio_{II}}

\newcommand{\Ratio}{\rho}
\newcommand{\RatioII}{\rho'}
\newcommand{\HFun}{\mathcal{V}}
\newcommand{\QFun}{\mathcal{Q}}
\newcommand{\lowestprice}{\PPrice}
\newcommand{\GFun}{G}
\newcommand{\WFun}{W}
\newcommand{\reserveinstanceapprox}{2}
\newcommand{\pricinginstanceapprox}{2.23}
\newcommand{\CDF}{F}
\newcommand{\BenchRev}{\textsc{ExAnteRev}}
\newcommand{\PriceRev}{\textsc{PriceRev}}
\newcommand{\OptPriceRev}{\textsc{OptPriceRev}}
\newcommand{\TriDistName}{triangular revenue curve}

\newcommand{\Instance}{\mathcal{I}}
\newcommand{\InstanceOf}[2]{\{{#1}_i\}_{i=1}^{\NumAgents}}
\newcommand{\TriInstanceOf}[3]{\{\TriDist({#1}_i, {#2}_i)\}_{i=1}^{#3}}

\newcommand{\RegInstanceSpace}{\textsc{Reg}}

\renewcommand{\RangeN}[1]{\RangeMN{1}{#1}}

\renewcommand{\secref}[1]{Section~\ref{#1}}

%% file: abstract.tex
\begin{abstract}
For selling a single item to agents with independent but
non-identically distributed values, the revenue optimal auction is
complex.  With respect to it, \citet{HR-09} showed that the
approximation factor of the second-price auction with an anonymous
reserve is between two and four.  We consider the more demanding problem of
approximating the revenue of the ex ante relaxation of the auction
problem by posting an anonymous price (while supplies last) and prove
that their worst-case ratio is $e$.  As a corollary, the upper-bound
of anonymous pricing or anonymous reserves versus the optimal auction
improves from four to $e$.  We conclude that, up to an $e$ factor,
discrimination and simultaneity are unimportant for driving revenue in
single-item auctions.
\end{abstract}

%% file: intro.tex
\section{Introduction}
\label{sec:introduction}
Methods from theoretical computer science are amplifying the
understanding of studied phenomena broadly.  A quintessential example
from auction theory is the following.  \citet{mye-81} is oft quoted as
showing that the second-price auction with a reserve price is revenue
optimal among all mechanisms for selling a single item.  This result
is touted as triumph for microeconomic theory as in practice
reserve-pricing mechanisms are widely prevalent, e.g., eBay's auction.
A key assumption in this result, though, is that the agents in the
auction are a priori identical; moreover, relaxation of this
assumption renders the theoretically optimal auction much more complex
and infrequently observed in practice.  Optimality of reserve pricing
with agent symmetry, thus, does not explain its prevalence broadly in
asymmetric settings, e.g., eBay's auction where agents can be
distinguished by public bidding history and reputation.  This paper
considers the approximate optimality of anonymous pricing and auctions
with anonymous reserves, e.g., eBay's buy-it-now pricing and auction,
and justifies their wide prevalence in asymmetric environments.

With two agents, anonymous reserve pricing is a tight two
approximation to the optimal auction; moreover, a surprising corollary
of a main result of \citet{HR-09} showed that the second-price auction
with anonymous reserves is generally no worse than a four
approximation.  The question of resolving the approximation factor
within $[2,4]$ has remained open for the last half decade.
Technically, (a) tight methods for understanding symmetric solutions
in asymmetric environments are undeveloped, and (b) the main method
for analyzing auction revenue is by Myerson's virtual values but for
this question virtual values give a mixed sign objective that renders
challenging the analysis of approximation.  The four approximation of
\citet{HR-09} employs the only known approach for resolving (b), an
approximate extension of the main theorem of \citet{BK-96}.  Our
approach directly takes on the challenge of (a) by giving a tight
analysis of anonymous pricing versus a standard upper bound (described
below); corollaries of this analysis are tightened upper bounds on
approximation of the optimal auction from four to $e \approx 2.718$
for both anonymous pricing and anonymous reserves.

In the Bayesian single-item auction problem agents' values are drawn
from a product distribution and expected revenue with respect to the
distribution is to be optimized.  Our development of the approximation
bound for anonymous pricing and reserves is based on the analysis of
four classes of mechanisms:
\begin{enumerate}
\item{\textbf{{Ex ante relaxation (a discriminatory pricing):}}} An ex ante pricing relaxes the feasibility
  constraint of the auction problem, from selling at most one item ex
  post, to selling at most one item in expectation over the draws of
  agents' values, i.e., ex ante.  Fixing a probability of serving a
  given agent the optimal ex ante mechanism offers this agent a posted
  price irrespective of the outcome of the mechanism for the other
  agents.  This relaxation was identified as a quantity of interest in
  \citet{CHK-07} and its study was refined by \citet{ala-11} and
  \citet{yan-11}.
  \vspace{-1mm}
\item{\textbf{{Auction}:}} An auction is any mechanism that maps
  values to outcome and payments subject to incentive and feasibility
  constraints.  The optimal auction was characterized by
  \citet{mye-81} and this characterization, though complex, is the
  foundation of modern auction theory.
  \vspace{-1mm}
\item{\textbf{{Anonymous reserve:}}} An anonymous reserve mechanisms
  is a variant of the second-price auction where bids below an
  anonymous reserve are discarded, the winner is the highest of the
  remaining agents, and the price charged is the maximum of the
  remaining agents' bids or the reserve if none other remain.
 \vspace{-1mm}
\item{\textbf{Anonymous pricing}}: An anonymous pricing mechanism posts an
  anonymous price and the first agent to arrive who is willing to pay
  this price will buy the item.
\end{enumerate}
\vspace{-1mm} For any distribution over agents' values the optimal
revenue attainable by each of these classes of mechanisms is
non-increasing with respect to the above ordering.  The final
inequality of optimal anonymous reserve exceeding optimal anonymous
pricing follows as with equal reserve and price, the former has only
higher revenue as competition drives a higher price.  The ex ante
relaxation is a quantity for analysis only, while the other problems
yield relevant mechanisms.

Our main technical theorem identifies the supremum over all instances
of the ratio of the revenues of the optimal ex ante relaxation to the
optimal anonymous pricing as the solution to an equation which
evaluates to $e$.  To our knowledge, this evaluation is not by any
standard progressions or limits that where previously known to
evaluate to $e$.  The theorem assumes that the distribution of agents'
values satisfies a standard {\em regularity} property that is
satisfied by common distributions, e.g., uniform, normal, exponential;
without this assumption we show that the approximation factor is $n$
for $n$-agent environments (see \cref{s:irregular}).



\begin{theorem}[\textbf {Anonymous pricing versus ex ante relaxation}]
\label{main-theorem}
\label{HQ:def}
For a single item environment with agents with independently (but
non-identically) distributed values from regular distributions, the
worst case approximation factor of anonymous pricing to the ex ante
relaxation is $\big(\HFun\left(\QFun^{-1}(1)\right)+1\big)$ which
evaluates to $e \approx 2.718$ where
\begin{align*}
 \HFun(\PPrice) &\triangleq \PPrice\cdot\ln\bigg(\frac{\PPrice^2}{\PPrice^2-1} \bigg),&
 \QFun(\PPrice) &\triangleq \int_{\PPrice}^{\infty} -\frac{\HFun'(\MaxValC)}{\MaxValC}\dif \MaxValC.
\end{align*}
\end{theorem}

Intuition for the theorem, as given by functions $\HFun(\cdot)$ and
$\QFun(\cdot)$, and its proof is as follows.  We write a mathematical
program to maximize the worst case approximation factor; a
tight-in-the-limit continuous relaxation of this program gives the
objective $1 + \HFun(\PPrice)$ subject to $\QFun(\PPrice) \leq 1$
which has the following interpretation. There is a continuum of agents
and each agent value distribution is given by a pointmass at a value
with some probability (and then a continuous distribution below the
pointmass to minimally satisfy the regularity property).  The function
$\HFun(\PPrice)$ is the expected pointmass value from agents with
pointmass value at least price $\PPrice$;\footnote{$\HFun(\cdot)$
  excludes the contribution from the ``highest valued agent'' which is
  1; hence the objective $1+\HFun(\PPrice)$.}  $\QFun(\PPrice)$ is the
expected number of these agents to realize their pointmass value.  The
optimal $\SPPrice$ meets the constraint with equality, i.e.,
$\QFun(\SPPrice) = 1$.

Corollaries of this theorem are the improved upper bounds by $e$ (from
$4$) on the worst-case approximation factor of anonymous reserves and
anonymous pricing with respect to the optimal auction.  On the worst
case instance of the theorem, however, the actual approximation factor
of anonymous reserve and anonymous pricing are
$\reserveinstanceapprox$ and $\pricinginstanceapprox$, respectively
(see \cref{appendix:sim}).  The latter improves on the known lower
bound of two; the former does not improve the known lower bound.  The
question of refining our understanding of the revenue of anonymous
reserves on worst-case instances and identifying a tight bound with
respect to the optimal auction remains open.  See
Figure~\ref{fig:gapbetweenmechanism}.

\begin{figure}[t]
\center
\begin{tikzpicture}[scale=1,transform shape]
\draw (3.5,0.5) node[font=\Large] {$\geq$};
\draw (7.5,0.5) node[font=\Large] {$\geq$};
\draw (11.5,0.5) node[font=\Large] {$\geq$};
\draw (0,0) rectangle (3,1);
\draw (1.5,0.5) node {Ex-ante Relax};
\draw (4,0) rectangle (7, 1);
\draw (5.5,0.5) node {Optimal Auction};
\draw (8,0) rectangle (11, 1);
\draw (9.5,0.5) node {$2^{nd}$Price+Resrv};
\draw (12,0) rectangle (15, 1);
\draw (13.5,0.5) node {Posted Price};
\draw (5.7,1) --++(0,0.3)--++(3.6,0)--++(0,-0.3);
\draw (7.5,1.6) node {$[2,e^*]$};
\draw (1.5,1) --++(0,1)--++(12,0)--++(0,-1);
\draw (7.5,2.25) node {$e^*$};
\draw (7,0) --++(1,-0.3)--++(5.3,0)--++(0,0.3);
\draw (9.5,-0.6) node {$[\pricinginstanceapprox^*,e^*]$};
\draw (1.7,0) --++(0,-0.3)--++(5.3,0)--++(1,0.3);
\draw (5.5,-0.6) node {$[2,e^*]$};
\end{tikzpicture}
\caption{Revenue gap between mechanisms of study; $^*$ denotes new bounds.}
\label{fig:gapbetweenmechanism}
\end{figure}
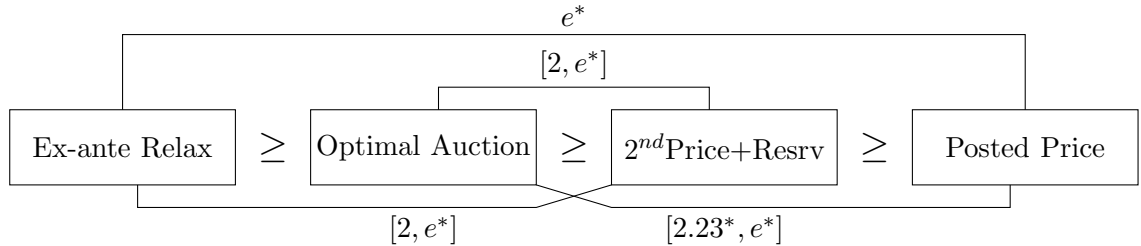

The corollary relating anonymous pricing to the optimal auction has
implications on mechanism design for agents with multi-dimensional
preferences \citep[e.g., for multiple items; cf.][]{CHK-07}.
Understanding these problems, though there has been considerable
recent progress, remains an area with fundamental open questions for
optimization and approximation.  Recently, \citet{HH-14} proved the
optimality of uniform pricing for a single unit-demand buyer with
values drawn from a large family of item-symmetric distributions.  An
immediate corollary of our anonymous pricing result is that, for a
unit-demand buyer with values drawn from an asymmetric product
distribution, uniform pricing is an $e$ approximation (improved from
four) to the optimal non-uniform pricing \citep[cf.][]{CD-11} and, via
a result of \citet{CMS-10}, a $2e$ approximation to the optimal
pricing over lotteries (i.e., randomized allocations, improved from
eight).  Further refinement of this latter bound remains an important
open question.  These approximation results for a single agent
automatically improve the approximation bounds for related multi-agent
mechanism design problems based on uniform pricing, e.g., from
\citet{AFHH-13}.  As one example, for selling an object that can be
configured on sale in one of $m$ configurations to $n$ agents with
independently (but non-identically) distributed values for each
configuration (also satisfying a regularity property), the
second-price auction with an anonymous reserve that configures the
object as the winner most prefers is a $2e^2 \approx 14.8$
approximation to the optimal auction (which is sometimes randomized;
improved from 32).

This work is part of a central area of study at the intersection of
computer science and economics that aims to quantify the performance
of simple, practical mechanisms versus optimal mechanisms (see
\citealp{har-13}, for a survey).  Immediately related results in this
area fit in to three broad categories, (i) anonymous and
discriminatory reserve pricing \citep{HR-09}, (ii) while-supplies-last
posted pricing \citep{CHMS-10}, and (iii) increased competition with
symmetric \citep{BK-96} and asymmetric agents \citep{HR-09}.  The four
approximation of \citet{HR-09} for anonymous reserve pricing is a
corollary of their result for (iii) on (i).  
In comparison this paper considers
(ii), foremost, and obtain a lower bound for (i) as a corollary.

The field of algorithmic mechanism design contains many questions of
constant approximation where tight bounds are not known.  A key
challenge of these problems is that the worst-case bounds are not
given by small instances, e.g., $n=2$ agents, but are instead
approached in the limit with $n$.  The field lacks general methods for
analysis of this kind of problem.  Our approach is similar to the
recent successful approach of \citet{CGL-14} which identified the
prior-free approximation ratio for digital good auctions as 2.42
(matching the lower bound of \citealp{GHKSW-06}). The approach writes
the approximation ratio as the value of a mathematical program.  In
both our problem and that of \citet{CGL-14} the worst-case instance is
attained in the limit with the number $n$ of agents.  
With two success stories for this approach in algorithmic mechanism
design, we are optimistic about the development of a set of tools for
analyzing worst case approximation factors and that these tools will
be useful for making progress on many other similar open questions in
the area.

%% file: prelim.tex
\section{Preliminaries and Notations}
\label{sec:prelim}

\paragraph{Revenue curves and regular instances.} 
Each agent $i$ has value drawn from a distribution $\CDF_i$ with
cumulative distribution function (CDF) denoted by $\CDF_i(\cdot)$.
The {\em revenue curve} $\Rev_i(\Prob) =
\Prob\cdot\CDF^{-1}_i(1-\Prob)$ gives the expected revenue obtained by
selling an item to agent $i$ with probability exactly $\Prob$, i.e.,
by posting price $\CDF^{-1}(1-\Prob)$.  The agent is {\em regular} if
its revenue curve $\Rev_i(\Prob)$ is concave in $\Prob$.  An
$\NumAgents$-agent instance $\Instance=\InstanceOf{\CDF}{\NumAgents}$
is regular if each agent's distribution is regular.  The family of all
regular instances for all $n\geq 1$ is denoted by $\RegInstanceSpace$.

\paragraph{Anonymous pricings.}
An anonymous pricing is a mechanism that posts a price $\Price$ that
is bought by an arbitrary agent whose value is at least the posted
price (if one exists).  The expected revenue of the anonymous pricing
$\Price$ for instance $\Instance=\InstanceOf{\CDF}{\NumAgetns}$ is
\begin{align}
    \PriceRev(\Instance, \Price) &\triangleq \Price\cdot \Big(1-\prod\nolimits_i \CDF_i(\Price)\Big).
\end{align}
The expected revenue of the optimal anonymous pricing is
\begin{align}
    \OptPriceRev(\Instance) &\triangleq \max_{\Price \in \PosReals}\ \PriceRev(\Instance, \Price).
\end{align}

\paragraph{Ex ante relaxation and optimal auctions.}
The revenue of the {\em ex ante relaxation}, which allocates to one
agent in expectation, gives an upper bound on the revenue of the
optimal auction.  For any instance $\Instance$, it can be easily
expressed in terms of the revenue curves of the agents.
\begin{alignat}{4}
    \BenchRev(\Instance)\triangleq&&\max&\quad & &\sum\nolimits_{i=1}^{\NumAgents} \Rev_i(\Prob_i) \label{eq:bench_rev} \\
    &&\text{subject to}& & &\sum\nolimits_{i=1}^{\NumAgents} \Prob_i \le 1 \notag \\
    &&               & & &\Prob_i \ge 0 & \quad&\forall i \in \RangeN{\NumAgents}  \notag.
\end{alignat}
Note: We would prefer to compare the performance of anonymous pricing
directly to the optimal auction of \citet{mye-81}; however, the
standard formulation of the expected revenue of the optimal mechanism
is difficult to analyze relative to the optimal anonymous pricing.  

\paragraph{Worst-case approximation ratio.}
The main task of this paper is to analyze the worst case ratio of the
revenue of the ex ante relaxation to the revenue of the optimal
anonymous pricing over all regular instances, that is
\begin{align}
    \Ratio \triangleq \sup_{\Instance\in\RegInstanceSpace}
    \frac{\BenchRev(\Instance)}{\OptPriceRev(\Instance)}
    \tag{P1}\label{eq:ratio}~,
\end{align}
where $\RegInstanceSpace$ denotes the space of all regular instances.

\paragraph{Triangular revenue curve instances.}
We will show that distributions with triangular-shaped revenue curves give worst case
instances for program \eqref{eq:ratio}.  A \emph{\TriDistName}
distribution, denoted $\TriDist(\MaxVal,\MaxProb)$ with parameters
$\MaxVal \in \Rangeab{0}{\infty}$ and $\MaxProb \in \RangeAB{0}{1}$,
has CDF given by
 \begin{align}
    \CDF(\PPrice) 
        &=  \begin{cases}
            1 & \PPrice \geq \MaxVal\\
            \frac{\PPrice\cdot(1-\MaxProb)}{\PPrice\cdot(1-\MaxProb)+\MaxVal\MaxProb} &  0 \le \PPrice < \MaxVal
            \end{cases} &&\forall \PPrice \in \PosReals.
\end{align}
The revenue curve corresponding to the above distribution has the form
of a triangle with vertices at $(0,0)$, $(\MaxProb,\MaxVal\MaxProb)$,
and $(1,0)$ as illustrated in \cref{fig:tri_rev_curve}; the revenue
curve's concavity implies that the distribution is regular. Note that
the CDF is discontinuous at $\MaxVal$ which corresponds to a pointmass
of $\MaxProb$ at value $\MaxVal$.  A {\em triangular revenue curve instance}
is given by $\Instance=\TriInstanceOf{\MaxVal}{\MaxProb}{\NumAgents}$
with $\sum_{i=1}^{\NumAgents} \bar{q}_i\leq 1$; with respect to it the revenue of
anonymous pricing $\Price$ and the ex ante relaxation are
given by
\begin{align}
  \PriceRev(\Instance, p) &= \PPrice\cdot\bigg(1-\prod\nolimits_{i : \MaxVal_i
      \ge \PPrice} \Big(1+\frac{\MaxVal_i\MaxProb_i}{\PPrice\cdot(1-\MaxProb_i)}\Big)^{-1} \bigg), \label{eq:tri_price_rev} \\
  \BenchRev(\Instance) &= \sum\nolimits_{i=1}^\NumAgents \MaxVal_i\MaxProb_i.   \label{eq:tri_bench_rev}
\end{align}

\begin{figure}[t]
\centering
\begin{minipage}[c]{3in}
\centering
\tikzstyle{letterlabel}=[thick,font=\fontsize{6}{6},color=black!70!blue]
\begin{tikzpicture}[scale=0.99, every node/.style={transform shape}]
\draw[->] (0,0)--+ (0,5);
\draw[->] (0,0)--+ (6,0);
\draw (-0.8,4.5) node {Revenue};
\draw (6.8,-0.5) node {Sale Probability};
\draw (1.5,1.2) node {slope $\MaxVal$};
\draw[dotted,<-] (0.8,1.95) --+ (0.5,-0.5);
\draw (-0.2,-0.2) node {$0$};
\draw (5.2,-0.2) node {$1$};
\draw (-0.3,4) node {$\MaxVal \MaxProb$};
\draw (1.5,-0.2) node {$\MaxProb$};
\draw (4.7,2) node {slope $\frac{\MaxVal\MaxProb}{1-\MaxProb}$};
\draw[dotted,<-] (4.2,1.2) --+ (0.5,0.5);
\draw (0,0) --+ (1.5,4) --+ (5,0);
\draw[dotted] (0,4)--++(1.5,0)--++(0,-4);
\end{tikzpicture}
\end{minipage}
\caption{Revenue curve of distribution $\TriDist(\MaxVal,\MaxProb)$.\label{fig:tri_rev_curve}}
\end{figure}

%% file: upperbound.tex
\section{Upper-Bound Analysis}
\label{sec:upperbound}%

Program \eqref{eq:ratio} defines a tight upper bound on the ratio,
denoted by $\Ratio$, of the revenue of the ex ante relaxation to the
revenue of the optimal anonymous pricing.  This program can be thought
of as a continuous optimization problem over regular distributions
with the objective of maximizing the aforementioned ratio.  The
current \lcnamecref{sec:upperbound} shows the upper bound of
\cref{main-theorem} while \secref{sec:lowerbound} shows tightness.

\paragraph{Overview of the analysis.}
By normalizing the optimal anonymous pricing revenue to be one,
$\eqref{eq:ratio}$ is equivalent to the following program:
\begin{alignat}{4}
    \Ratio=&& \sup_{\Instance\in\RegInstanceSpace}&\quad & &\BenchRev(\Instance) & \qquad & \tag{P2} \label{eq:alt_bound}\\
    &&\text{subject to}& & &\PriceRev(\Instance, \PPrice) \le 1 &&\forall \PPrice \ge 1. \tag{P2.1}\label{eq:price_cons}
\end{alignat}
Note that $\PriceRev(\Instance, \PPrice)< 1$ for $\PPrice\in[0,1)$, so it is safe to assume prices are in range $[1,+\infty)$. We show that for any fixed $\NumAgents$ the supremum of this program
is approached even when restricting to triangular revenue curve
instances, i.e., ones of the form
$\{\TriDist(\MaxVal_i,\MaxProb_i)\}_{i=1}^\NumAgents$ with $\sum_i
\MaxProb_i \le 1$ as defined in \secref{sec:prelim}.  Consequently, the problem is reduced to a discrete
optimization problem over variables
$\MaxVals\triangleq(\MaxVal_1,\ldots,\MaxVal_n)$ and
$\MaxProbs\triangleq(\MaxProb_1,\ldots,\MaxProb_n)$. An \emph{assignment} for this optimization problem refers to a pair $(\MaxVals,\MaxProbs)$. This optimization
problem is still of infinite dimension because $\NumAgents$ is itself
a variable. It also turns out to be highly non-convex. Re-index
$\MaxVals$ such that $\MaxVal_1\ge \ldots \ge \MaxVal_\NumAgents$. We
will show that, for any fixed $\NumAgents$,
inequality~\eqref{eq:price_cons} can be assumed without loss of
generality to be tight for all $\PPrice \in
\{\MaxVal_1,\ldots,\MaxVal_\NumAgents\}$; otherwise an instance for which at least one of these constraints is not tight
could be modified to make all these constraints tight while improving
the objective. Thus,
\begin{align}
  \PriceRev(\TriInstanceOf{\MaxVal}{\MaxProb}{\NumAgents}, \MaxVal_k) &= 1  && \forall k \in \RangeN{\NumAgents}.
  \label{eq:tight}
\end{align}
Observe that for each $k \in \RangeN{\NumAgents}$ the left hand side
of the above equation only depends on the first $k$ agents, because
the the valuations of the rest of the agents are always below
$\MaxVal_k$. Consequently once $\MaxVal_1,\ldots,\MaxVal_\NumAgents$
are fixed, we can compute $\MaxProb_1,\ldots,\MaxProb_\NumAgents$ by
solving equation \eqref{eq:tight} for $k \in \RangeN{\NumAgents}$ and
using forward substitution. Unfortunately, the resulting formulation
of $\MaxProb_k$ in terms of $\MaxVal_1, \ldots, \MaxVal_k$ is
analytically intractable for $k \ge 2$. To work around this
intractability issue, we relax inequality \eqref{eq:price_cons} in
such a way that it leads to a tractable formulation of $\MaxProbs$ in
terms of $\MaxVals$. We also show that the relaxed inequality is tight
which implies the value of the relaxed program is equal to that of the
original program. Finally, we show that the supremum of the relaxed
program is attained when $\NumAgents \to \infty$, and roughly speaking
the instance converges to a continuum of infinitesimal agents with
triangular revenue curve distributions. For this continuum of agents,
$\Ratio$ is given simply by the optimization of $\PPrice$ in the
objective $1 + \HFun(\PPrice)$ subject to the constraint $\QFun(\PPrice)
\leq 1$ for the two functions $\HFun(\cdot)$ and $\QFun(\cdot)$ given
in the statement of \cref{main-theorem}.


\paragraph{Reduction to triangular revenue curve instances.}

We begin by showing that without loss of generality we can restrict
program \eqref{eq:alt_bound} to triangle revenue curve instances.

\begin{lemma}
\label{res:tri_dist}%
The supremum of program \eqref{eq:alt_bound} is approached by triangle
revenue curve instances, i.e., of the form
$\hat{\Instance}=\TriInstanceOf{\MaxVal}{\MaxProb}{\NumAgents}$ with
$\sum_{i=1}^\NumAgents{\MaxProb_i} \le 1$.
\end{lemma}
\begin{proof}

We will show that for any regular instance
$\Instance=\InstanceOf{\CDF}{\NumAgents}$, there exists a
corresponding instance
$\hat{\Instance}=\TriInstanceOf{\MaxVal}{\MaxProb}{\NumAgents}$ with
$\sum_{i=1}^\NumAgents{\MaxProb_i} \le 1$ yielding the same optimal ex
ante revenue and (weakly) smaller expected revenue from the optimal
anonymous price.


Let $\MaxProbs$ be an optimal assignment for the ex ante relaxation
program~\eqref{eq:bench_rev} that computes $\BenchRev(\Instance)$. Set
$\MaxVal_i\leftarrow \Rev_i(\MaxProb_i)/\MaxProb_i$ for each $i\in
\RangeN{\NumAgents}$, where $\Rev_i$ is the revenue curve of
$\CDF_i$. We show changing agent $i$'s valuation distribution to
$\TriDist(\MaxVal_i,\MaxProb_i)$ can only decrease the revenue of any
anonymous pricing ($\PriceRev$) while preserving the revenue of the ex ante
relaxation ($\BenchRev$), which implies the statement of the
\lcnamecref{res:tri_dist}.

Let $\hat{\Rev}_i$ be the revenue curve of $\TriDist(\MaxVal_i,\MaxProb_i)$. Observe that the change of distributions does
not affect $\BenchRev$ because $\hat{\Rev}_i(\MaxProb_i)=\Rev_i(\MaxProb_i)$ for all $i\in \RangeN{\NumAgents}$, and
$\hat{\Rev}_i$ is a lower bound on $\Rev_i$ elsewhere as $\Rev_i$ is concave (see \cref{fig:tri_dist_transform}). Therefore the replacement
preserves the optimal value of convex program of \cref{eq:bench_rev} which implies
$\BenchRev(\hat{\Instance})=\BenchRev(\Instance)$.
%
%
%
%

\begin{figure}[ht]
  \begin{minipage}[b]{0.45\linewidth}
    \centering
\tikzstyle{letterlabel}=[thick,font=\fontsize{6}{6},color=black!70!blue]
\begin{tikzpicture}[scale=0.99, every node/.style={transform shape}]
\draw[->] (0,0)--+ (0,5);
\draw[->] (0,0)--+ (6,0);
\draw (-0.8,4.8) node {Revenue};
\draw (4.5,-0.7) node {Sale Probability};
\draw (-0.2,-0.2) node {$0$};

\draw (5.2,-0.2) node {$1$};
\draw (-0.6,3.85) node {$\Rev_i(\MaxProb_i)$};
\draw (1.3,-0.2) node {$\MaxProb_i$};
\draw (4.6, 4.6) node {$\Rev_i$: black, $\hat{\Rev}_i$: red.};
\draw[dotted] (1.3,0)--++(0,3.85)--++(-1.3,0);
\draw  plot [smooth] coordinates {(0,0)(1,3.5)(2,4)(3,3.4) (5,0)};

\draw [red,dashed] (0,0) --+ (1.3,3.85) --+ (5,0);
\end{tikzpicture}
\caption{Replacing regular distribution $\CDF_i$ with triangular revenue curve distribution $\TriDist(\MaxVal_i,\MaxProb_i)$.\label{fig:tri_dist_transform}}
    \vspace{4ex}
  \end{minipage}
  \hspace{0.1\linewidth}
  \begin{minipage}[b]{0.45\linewidth}
    \centering
\begin{tikzpicture}[scale=1,transform shape]
\draw[->] (0,0)--+ (0,5);
\draw[->] (0,0)--(6,0);
\draw (5,-0.3) node {$1$};
\draw (4.5,-0.7) node {Sale Probability};
\draw (-0.2,-0.2) node {$0$};
\draw (0,0) to[out=80,in=180] (2,4) to[out=0,in=100] (5,0);
\draw[red,dashed] (0,0) -- (3,5);
\draw (3.2,5.2) node {$\Price$};
\draw[dotted] (2.38,3.95) -- + (0,-3.95);
\draw[dotted] (2.38,3.95)--+(-2.38,0);
\draw (-0.4,3.95) node {$\Rev(\Prob)$};
\draw (-0.8,4.8) node {Revenue};
\draw (2.38, -0.3) node {$\Prob$};
\end{tikzpicture}
\caption{Intersection of revenue curve and price line $\Price$.\label{fig:rev_curve}}\vspace{0.5cm}
    \vspace{4ex}
  \end{minipage}
\end{figure}

Next, we show the replacement may only decrease the value of $\PriceRev(\Price)$ at any $\Price>0$. Fix a price
$\Price$, and consider the price line corresponding to $\Price$, that is, the line with slope $\Price$ passing through
the origin (see \cref{fig:rev_curve}). Observe that the probability of agent $i$'s valuation being above $\Price$ is
equal to the $\Prob$ at which $\Rev_i(\Prob)$ intersects price line $\Price$. Given that $\hat{\Rev}_i$ is a lower
bound on $\Rev_i$ everywhere, the replacement may only  decrease the probability of agent $i$'s valuation being above
$\Price$. Consequently, given that agents' valuations are distributed independently, the replacement may only decrease
the revenue from sale at any anonymous price $\Price$, which implies $\PriceRev(\hat{\Instance}) \le
\PriceRev(\Instance)$.
\end{proof}


Combining \cref{res:tri_dist} with the formulation of $\PriceRev$ and $\BenchRev$ from
\cref{eq:tri_bench_rev,eq:tri_price_rev}  yields the following non-convex program for computing $\Ratio$:

\begin{alignat}{4}
\Ratio= && \sup_{\NumAgents\in\PosInts,\MaxVals,\MaxProbs}    & \qquad &  \sum_{i=1}^\NumAgents \MaxVal_i\MaxProb_i \qquad\qquad\qquad & & \qquad & \tag{P3} \label{eq:tri_ratio} \\
&&\text{subject to}   & & \PPrice\cdot\left(1-\prod_{i : \MaxVal_i \ge \PPrice} \frac{1}{1+\frac{\MaxVal_i\MaxProb_i}{\PPrice\cdot(1-\MaxProb_i)}} \right) &\le 1,
                        & & \forall \PPrice \geq 1 \tag{P3.1}
                        \label{eq:tri_ratio:price_cons} \\
&&                    & &         \sum_{i=1}^{\NumAgents} \MaxProb_i &\le 1 \notag \\
&&                    & &         \MaxVal_i &\ge 0, \MaxProb_i \ge 0 & &
\forall i \in \RangeN{\NumAgents}. \notag
\end{alignat}

\paragraph{Relaxations and canonical assignments.}
In this section we find a relaxation of program~\eqref{eq:tri_ratio} where
the corresponding {\em pricing revenue constraint}
\eqref{eq:tri_ratio:price_cons} is tight for all $\PPrice \in
\{\MaxVal_1,\ldots,\MaxVal_\NumAgents\}$ and can thus be written as a program
on variables $\MaxVals$ alone (i.e., by solving for the appropriate
$\MaxProbs$ in terms of $\MaxVals$).  To simplify the solution of $\MaxProbs$
in terms of $\MaxVals$, we will first make a series of relaxations to the
pricing revenue constraint \eqref{eq:tri_ratio:price_cons}. We will point which of these relaxations are obviously tight, the others we will prove to be tight in the limit with the number of agents
$\NumAgents$ in \cref{sec:lowerbound}, where we derive the matching lower bound. 

These
relaxations will turn out to be tight in the limit 


\Cref{res:ratio2} formalizes these relaxations as sketched below, the formal proof is given in \cref{proof:res:hbound}.
First, observe that the pricing revenue constraint
\eqref{eq:tri_ratio:price_cons} can be rearranged as
\begin{align*}
  \prod_{i : \MaxVal_i \ge \PPrice} \left(1+\frac{\MaxVal_i \MaxProb_i}{\PPrice\cdot(1-\MaxProb_i)}\right) &\le
      \left(\frac{\PPrice}{\PPrice-1}\right) & & \forall \PPrice \geq 1.
      \end{align*}
The first relaxation drops the constraint on $\PPrice \not\in
\{\MaxVal_1, \ldots, \MaxVal_\NumAgents\}$; this is without loss as
the optimal anonymous price is always in $\{\MaxVal_1, \ldots,
\MaxVal_\NumAgents\}$.  We re-index such that $\MaxVal_1 \ge \cdots
\ge \MaxVal_\NumAgents$ and rephrase the relaxed constraint as
    \begin{align*}
  \prod_{i=1}^k\left(1+\frac{\MaxVal_i \MaxProb_i}{\MaxVal_k\cdot(1-\MaxProb_i)}\right) &\le
      \left(\frac{\MaxVal_k}{\MaxVal_k-1}\right) & & \forall k \in \RangeN{\NumAgents}.
     \end{align*}
As the second relaxation we drop the term $(1-\MaxProb_i)$ from the denominator of the left hand side and
    take the logarithm of both sides to get
    \begin{align*}
  \sum_{i=1}^k \ln\left(1+\frac{\MaxVal_i \MaxProb_i}{\MaxVal_k}\right) &\le
      \ln\left(\frac{\MaxVal_k}{\MaxVal_k-1}\right) & & \forall k \in \RangeN{\NumAgents}.
     \end{align*}
As the third relaxation we upper-bound $\ln\left(1+\frac{\MaxVal_i \MaxProb_i}{\MaxVal_k}\right)$ by
  $\frac{1}{\MaxVal_k}\ln(1+\MaxVal_i\MaxProb_i)$ for $i\ge 2$. Rearranging gives
 \begin{align*}
     \sum_{i=2}^k \ln\left(1+\MaxVal_i \MaxProb_i\right) &\le
    \MaxVal_k\ln\left(\frac{\MaxVal_k^2}{(\MaxVal_k-1)(\MaxVal_k+\MaxVal_1\MaxProb_1)}\right) & & \forall k \in \RangeMN{2}{\NumAgents}.
\end{align*}
The previous relaxation uses the fact that $\frac{1}{a}\ln(1+b) \le \ln(1+\frac{b}{a})$ for all $a \ge 1,b \ge 0$.  In the proof, we will
also show that $\MaxVal_1\MaxProb_1$ can be replaced with $1$ both in the above constraint and in the objective function
without loss of generality. Putting everything together, we will obtain the following program as a relaxation of program
\eqref{eq:tri_ratio}.

\begin{alignat}{4}
\RatioII= &&\sup_{\NumAgents\in \PosInts,\MaxVals,\MaxProbs} & \qquad &  1+\sum_{i=2}^{\NumAgents} \MaxVal_i \MaxProb_i
    & & \qquad & \tag{P4}\label{eq:tri_ratio2} \\
&&\text{subject to}   & &         \sum_{i=2}^k \ln\left(1+\MaxVal_i \MaxProb_i\right) &\le \HFun(\MaxVal_k)
    & & \forall k \in \RangeMN{2}{\NumAgents} \label{firstconstraintP4} \tag{P4.1}\\
&&                    & &         \sum_{i=2}^{\NumAgents} \MaxProb_i &\le 1 \notag \\
&&                    & &         \MaxVal_{i+1} &\le \MaxVal_i,
    & & \forall i \in \RangeMN{2}{\NumAgents-1} \notag \\
&&                    & &         \MaxVal_i &\ge 0, \MaxProb_i \ge 0  & &
\forall i \in \RangeN{\NumAgents} \notag.
\end{alignat}
\noindent where $\HFun(\cdot) = \PPrice\cdot\ln\big(\frac{\PPrice^2}{\PPrice^2-1}\big)$ is defined in \cref{HQ:def}.
\begin{lemma}
The value of program \eqref{eq:tri_ratio2}, denoted by $\RatioII$, is an upper bound on the value of program
\eqref{eq:tri_ratio} which is $\Ratio$.
\label{res:ratio2}
\end{lemma}

Next we show that we can assume without loss of generality the pricing
revenue constraint \eqref{firstconstraintP4} is tight for all $k \in
\RangeMN{1}{\NumAgents}$ in program \eqref{eq:tri_ratio2}. That will
allow us to specify one set of variables (e.g., $\MaxProbs$) in terms
of the other set of variables (e.g., $\MaxVals$), which consequently
allows us to eliminate the former variables and drop the pricing revenue
constraint (\ref{firstconstraintP4}). To this end, we first define a
\emph{canonical} feasible solution for \eqref{eq:tri_ratio2},
restriction to which is without loss given by \cref{res:tight}.

\begin{definition}
A feasible assignment $(\MaxVals,\MaxProbs)$ for \eqref{eq:tri_ratio2} is {\em canonical} if the pricing constraint (\ref{firstconstraintP4})
is tight for all $k \in \RangeMN{2}{\NumAgents}$.
\end{definition}
\begin{lemma}
\label{res:tight}%
For any feasible assignment $(\MaxVals,\MaxProbs)$ for \eqref{eq:tri_ratio2}, there exists an equivalent canonical
feasible assignment $(\MaxVals',\MaxProbs')$ obtaining the same objective value, that is $\sum_i \MaxVal_i\MaxProb_i = \sum_i
\MaxVal'_i\MaxProb'_i$.
\end{lemma}
\begin{proof}
Without loss of generality assume $\MaxProb_k > 0$ for all $k \in \RangeMN{2}{\NumAgents}$.\footnote{If
$\MaxProb_k=0$, we can drop agent $k$ without affecting feasibility or objective value.} The right hand side of the pricing constraint (\ref{firstconstraintP4}) is $\HFun(\MaxVal_k)$ which is decreasing in $\MaxVal_k$ (see~\cref{res:mono}) and approaches $0$ as
$\MaxVal_k\to\infty$, so for every $k \in \RangeMN{2}{\NumAgents}$ there exists $\MaxVal'_k \ge \MaxVal_k$ such that
\begin{align*}
    \sum_{i=2}^k \ln\left(1+\MaxVal_i \MaxProb_i\right) &= \HFun(\MaxVal'_k) & & \forall k \in \RangeMN{2}{\NumAgents}.
\end{align*}
Observe that by the above construction we always have $\MaxVal'_2 \ge \ldots \ge \MaxVal'_\NumAgents$. We then decrease  $\MaxProb_k$ to
$\MaxProb'_k = \MaxProb_k \frac{\MaxVal_k}{\MaxVal'_k}$ for each $k \in \RangeMN{2}{\NumAgents}$ to obtain the desired
assignment $(\MaxVals', \MaxProbs')$.
\end{proof}

By \cref{res:tight}, we can restrict our attention to canonical assignments of \eqref{eq:tri_ratio2} without loss of
generality. In particular, we can fully identify such a canonical assignment  by specifying only
$\MaxVals=(\MaxVal_1,\ldots, \MaxVal_\NumAgents)$ since the corresponding $\MaxProbs$ is given by
\begin{align}
    \MaxProb_k &= \frac{e^{\HFun(\MaxVal_k)-\HFun(\MaxVal_{k-1})}-1}{\MaxVal_k} && \forall k \in \RangeMN{2}{\NumAgents}. \label{eq:max_prob}
\end{align}


Therefore we can obtain from program \eqref{eq:tri_ratio2} the
following program.
\begin{alignat}{4}
\RatioII = &&\sup_{\NumAgents\in\PosInts,\MaxVals} & \qquad &  1&+\sum_{i=2}^{\NumAgents} \MaxVal_i \MaxProb_i
     & \qquad & \tag{P5}\label{eq:ratio3} \\
&&\text{subject to}   & &
    \MaxProb_k &= \frac{e^{\HFun(\MaxVal_k)-\HFun(\MaxVal_{k-1})}-1}{\MaxVal_k} && \forall k \in \RangeMN{2}{\NumAgents} \tag{P5.1}\\
&&                    & &         \sum_{i=2}^{\NumAgents} \MaxProb_i &\le 1 \label{eq:capacity} \tag{P5.2}\\
&&                    & &         \MaxVal_{i+1} &\le \MaxVal_i,
    & & \forall i \in \RangeMN{2}{\NumAgents-1} \notag\\
&&                    & &         \MaxVal_i &\ge 0,~ \MaxProb_i \ge 0 & &
\forall i \in \RangeN{\NumAgents} \notag.
\end{alignat}

%
\paragraph{Continuum of agents.}
Given that $\NumAgents$ itself is a variable, a solution to program
\eqref{eq:ratio3} can be practically specified by a finite subset
$\MaxVals \subset \PosReals$ where $\MaxVal_i$ is the $i$th largest
value in that subset.  We now show that the optimal solution to
program \eqref{eq:ratio3} corresponds to
$\MaxVals=\RangeAb{\lowestprice}{\infty}$ (for some $\lowestprice >
1$) which can be viewed as an instance with infinitely many
infinitesimal agents.

For any given $\PPrice' >\PPrice > 1$, we define a continuum of agents $\RangeAb{\PPrice}{\PPrice'}$ by defining for each $m \in
\PosInts$ a discrete family of agents of size $m$ spanning $\RangeAb{\PPrice}{\PPrice'}$ and by taking the limit of this family
as $m \to \infty$. Formally, for each $m \in \PosInts$, we consider the family of agents with distributions $
\{\TriDist(\MaxValD_j, (e^{\HFun(\MaxValD_j)-\HFun(\MaxValD_{j-1})}-1)/\MaxValD_j)\}_{j=1}^{m}$ where
$\MaxValD_j=\PPrice'+\frac{j}{m}(\PPrice-\PPrice')$. Observe that the agents in these families satisfy equation~\eqref{eq:max_prob}.
Furthermore, observe that
\begin{align*}
  \lim_{\delta \to 0} \left ( \frac{e^{\HFun(\MaxValC)-\HFun(\MaxValC+\delta)}-1}{\MaxValC}
    \cdot \frac{1}{\delta}\right )&= -\frac{\HFun'(\MaxValC)}{\MaxValC}.
\end{align*}
Therefore in a continuum of agents $\RangeAb{\PPrice}{\PPrice'}$ each infinitesimal agent $\MaxValC \in
\RangeAb{\PPrice}{\PPrice'}$ has a distribution of $\TriDist(\MaxValC, -\frac{\HFun'(\MaxValC)}{\MaxValC}\,\dif \MaxValC),$
which implies that the contribution of $\RangeAb{\PPrice}{\PPrice'}$ to the objective value of \eqref{eq:ratio3} is
\begin{align}
    \int_{\PPrice}^{\PPrice'} \MaxValC\cdot(-\tfrac{\HFun'(\MaxValC)}{\MaxValC})\dif \MaxValC
        &= \HFun(\PPrice)-\HFun(\PPrice'),
    \label{eq:h_contrib}
\end{align}
and the contribution of $\RangeAb{\PPrice}{\PPrice'}$ to the left hand side of the constraint  \eqref{eq:capacity}, i.e.\@  $\sum_i \MaxProb_i \le 1$
which is referred to as the \emph{capacity constraint}, is
\begin{align}
    \int_{\PPrice}^{\PPrice'} -\tfrac{\HFun'(\MaxValC)}{\MaxValC}\dif \MaxValC &= \QFun(\PPrice)-\QFun(\PPrice'),
    \label{eq:q_contrib}
\end{align}
where $\QFun(\MaxValC) = \int_{\PPrice}^{\infty} -\frac{1}{\MaxValC}\HFun'(\MaxValC)\dif \MaxValC$ as defined in \cref{HQ:def}.

Via the above derivation of a continuum of agents,
program~\eqref{eq:ratio3}, on the instance corresponding to the
continuum $[\PPrice,\infty)$, simplifies as:
\begin{alignat}{6}
\Ratio''=&& \max_{\PPrice\geq 1}&\quad & &1+\HFun(\PPrice) & \qquad & \tag{P6}\label{eq:ratio6} \\
    &&\textrm{subject to}& & &\QFun(\PPrice)\leq 1.\notag && .
\end{alignat}


Next we will sketch a construction that demonstrates that any feasible
solution $\MaxVals=(\MaxVal_1,\ldots,\MaxVal_\NumAgents)$ to
program~\eqref{eq:ratio3} can be replaced by a continuum of agents
that corresponds to an interval $\RangeAb{\lowestprice}{\infty}$ (for
some $\lowestprice > 1$ to be determined) and the objective of
\eqref{eq:ratio3} is strictly increased.
Note that $\MaxVal_1$ does not appear anywhere in \eqref{eq:ratio3}; for
notational convenience we redefine it as $\MaxVal_1=\infty$. Suppose for each
$i \in \RangeMN{2}{\NumAgents}$ we replace the agent $\MaxVal_i$ with the
continuum of agents $\RangeAb{\MaxVal_i}{\MaxVal_{i-1}}$. It follows from
\cref{eq:h_contrib,eq:max_prob} that this replacement changes the object
value of \eqref{eq:ratio3} by
$\HFun(\MaxVal_i)-\HFun(\MaxVal_{i-1})-\MaxVal_i\MaxProb_i=\ln(1+\MaxVal_i\MaxProb_i)-\MaxVal_i\MaxProb_i
< 0$ which is unfortunately always negative and thus the opposite of what we
want to prove. On the other hand, it follows from
\cref{eq:q_contrib,eq:max_prob} that this replacement also changes the left
hand side of the capacity constraint \eqref{eq:capacity} by
$\QFun(\MaxVal_i)-\QFun(\MaxVal_{i-1})-\MaxProb_i$ which is also negative (as
we will show later), and thus creates some slack in the capacity
constraint~\eqref{eq:capacity}.  Summing over the slack created in the
capacity constraint~\eqref{eq:capacity} from converting each agent to a
continuum, we can add a new continuum of agents
$\RangeAb{\lowestprice}{\MaxVal_n}$ where $\lowestprice < \MaxVal_n$ is
chosen to make the capacity constraint \eqref{eq:capacity} tight. As a
consequence of the following claims, the net change in the objective value
from this transformation is positive.
\begin{enumerate}[(i)]
\item
\label{it:1}%
The amount of slack created in the capacity constraint \eqref{eq:capacity} by replacing $\MaxVal_i$ with
$\RangeAb{\MaxVal_i}{\MaxVal_{i-1}}$ is more than the decrease in the objective value of \eqref{eq:ratio3}. Using
~\cref{eq:h_contrib,eq:q_contrib}, we can formally write this claim as
$\MaxProb_i-(\QFun(\MaxVal_i)-\QFun(\MaxVal_{i-1})) > \MaxVal_i\MaxProb_i-(\HFun(\MaxVal_i)-\HFun(\MaxVal_{i-1}))$.
This is proved below in \cref{res:twist}.
\item
\label{it:2}%
If there is a slack of $\Delta > 0$ in the capacity constraint \eqref{eq:capacity}, it can be used to extend the last continuum of agents
to increase the objective value by more than $\Delta$. Using ~\cref{eq:h_contrib,eq:q_contrib}, we can formalize this
claim as follows: if $\lowestprice$ is chosen such that $\QFun(\lowestprice)-\QFun(\MaxVal_n)=\Delta$, then
$\HFun(\lowestprice)-\HFun(\MaxVal_n) > \Delta$. This is proved below in \cref{res:mono}.
\end{enumerate}

The suggestion from the above construction is that from any solution
$\MaxVals$ to program~\eqref{eq:ratio3}, a price $\lowestprice$ can be
identified such that the continuum of agents on
$[\lowestprice,\infty)$ has higher objective value.  In other words,
  the optimal values of program~\eqref{eq:ratio3} and
  program~\eqref{eq:ratio6} are equal.  This is proved below in
  \cref{lastprog:lemma}; though we defer the
  proof that the solution of program \eqref{eq:ratio6} corresponds to
  a limit solution of program~\eqref{eq:ratio3} to \cref{sec:lowerbound}.

\paragraph{Algebraic upper-bound proof.}
The rest of this \lcnamecref{sec:upperbound} develops a formal but
purely algebraic proof that is based on the approach sketched in the
previous paragraphs.
%
The proofs of the first two lemmas, below,
can be found in \cref{proof:res:mono}.

\begin{lemma}
\label{res:mono}%
The functions $\HFun(\PPrice)$, $\QFun(\PPrice)$, and $\HFun(\PPrice)-\QFun(\PPrice)$ are all decreasing in $\PPrice$,
for $\PPrice > 1$.
\end{lemma}

\begin{lemma}
\label{res:hbound}%
for any $\PPrice'> \PPrice > 1$ the following inequality holds: $
\HFun(\PPrice)-\HFun(\PPrice') <
\ln(\frac{\PPrice}{\PPrice-1})-\ln(\frac{\PPrice'}{\PPrice'-1})$ .
\end{lemma}

\begin{lemma}
\label{res:twist}%
For any $\PPrice' > \PPrice > 1$ and $\Prob = \frac{e^{\HFun(\PPrice)-\HFun(\PPrice')}-1}{\PPrice}$ the following
inequalities hold:
\begin{align}
    \Prob-(\QFun(\PPrice)-\QFun(\PPrice')) &\ge \PPrice\Prob -(\HFun(\PPrice)-\HFun(\PPrice')) \ge 0
\end{align}
\end{lemma}
\begin{proof}
Define
\begin{align*}
    \WFun(\PPrice,\PPrice') &\triangleq \HFun(\PPrice)-\QFun(\PPrice)-\HFun(\PPrice')+\QFun(\PPrice')+\Prob-\PPrice\Prob \\
                            &= \HFun(\PPrice)-\QFun(\PPrice)-\HFun(\PPrice')+\QFun(\PPrice')
                                -(\PPrice-1)\big(e^{\HFun(\PPrice)-\HFun(\PPrice')}-1\big)/{\PPrice}.
\end{align*}
Observe that proving the first inequality in the statement of the lemma is equivalent to proving
$\WFun(\PPrice,\PPrice')> 0$. We instead prove that $\WFun(\PPrice,\PPrice')$ is increasing in $\PPrice'$ which
together with the trivial fact that $\WFun(\PPrice,\PPrice) = 0$ implies $\WFun(\PPrice,\PPrice') > 0$.
\vspace{-2mm}
\begin{align*}
    \pd{}{\PPrice'}\WFun(\PPrice,\PPrice')
                &=- \HFun'(\PPrice')+\HFun'(\PPrice')/{\PPrice'}
                                +(\PPrice-1)\,\HFun'(\PPrice') \, {e^{\HFun(\PPrice)-\HFun(\PPrice')}}/{\PPrice} \\
                &= -\HFun'(\PPrice')\left[\frac{\PPrice'-1}{\PPrice'}-\frac{\PPrice-1}{\PPrice}\, e^{\HFun(\PPrice)-\HFun(\PPrice')}\right] \\
                &> -\HFun'(\PPrice')\left[\frac{\PPrice'-1}{\PPrice'}
                    -\frac{\PPrice-1}{\PPrice}\, e^{\ln(\frac{\PPrice}{\PPrice-1})-\ln(\frac{\PPrice'}{\PPrice'-1})}\right]=0.
\end{align*}
The final inequality follows from \cref{res:mono,res:hbound}:
by \cref{res:mono}, $-\pd{}{\PPrice'} \HFun(\PPrice') > 0$; and by \cref{res:hbound},
$e^{\HFun(\PPrice)-\HFun(\PPrice')}$ is less than
$e^{\ln(\frac{\PPrice}{\PPrice-1})-\ln(\frac{\PPrice'}{\PPrice'-1})}$, so replacing the former with the latter only
decreases the value of the expression inside the brackets because its coefficient is $-\frac{\PPrice-1}{\PPrice}$ which
is negative.

The second inequality in the statement of the lemma follows trivially
from the fact that
$\HFun(\PPrice)-\HFun(\PPrice')=\ln(1+\PPrice\Prob)$ thus $
\PPrice\Prob -(\HFun(\PPrice)-\HFun(\PPrice')) = \PPrice\Prob
-\ln(1+\PPrice\Prob) > 0$.
\end{proof}
\vspace{-3mm}

\begin{lemma}
\label{lastprog:lemma}
The value of program \eqref{eq:ratio6}, denoted by $\Ratio''$, is an
upper bound on the value of program \eqref{eq:ratio3} which is
$\Ratio'$.
\end{lemma}
\begin{proof}
Let $(\MaxVals,\MaxProbs)$ be any arbitrary feasible assignment for
program \eqref{eq:ratio3}. We show there exists a feasible assignment
for program \eqref{eq:ratio6} with objective value upper bounding the
objective value of $(\MaxVals,\MaxProbs)$ in program
\eqref{eq:ratio3}.  Define $\SPPrice\triangleq\QFun^{-1}(1)$, a
candidate solution to program~\eqref{eq:ratio6} that meets the
feasibility constraint with equality. Note that such a $\SPPrice$
exists because $\QFun(\infty)=0$, $\QFun(1)=\infty$, and
$\QFun(\cdot)$ is continuous.  Observe that the objective value of
\eqref{eq:ratio3} for $(\MaxVals,\MaxProbs)$ satisfies:
\vspace{-2mm}
\begin{align*}
    1+\sum_{k=2}^\NumAgents \MaxVal_k\MaxProb_k
        &\le 1+\sum_{k=2}^\NumAgents \left(\HFun(\MaxVal_k)-\HFun(\MaxVal_{k-1})-(\QFun(\MaxVal_k)-\QFun(\MaxVal_{k-1}))+\MaxProb_k\right)
            &&\text{by \cref{res:twist} and $\MaxVal_1=\infty$ } \\
        &= 1+\HFun(\MaxVal_\NumAgents)-\QFun(\MaxVal_\NumAgents)+\sum_{k=2}^\NumAgents \MaxProb_k && \text{as $\HFun(\infty)=\QFun(\infty)=0$}\\
        &\le 1+\HFun(\MaxVal_\NumAgents)-\QFun(\MaxVal_\NumAgents)+1 &&\text{as $\sum_{k=2}^\NumAgents \MaxProb_k \le 1$} \\
        &< 1+\HFun(\SPPrice)-\QFun(\SPPrice)+1 &&\text{as proved below} \tag{$*$} \label{eq:dagger}\\
        &= 1+\HFun(\SPPrice) &&\text{as $\QFun(\SPPrice)=1$}.
\end{align*}
To prove inequality \eqref{eq:dagger} we show that
$\SPPrice<\MaxVal_\NumAgents$ which together with \cref{res:mono}
implies $\HFun(\SPPrice)-\QFun(\SPPrice) >
\HFun(\MaxVal_\NumAgents)-\QFun(\MaxVal_\NumAgents)$. To prove
$\SPPrice<\MaxVal_\NumAgents$ observe that \cref{res:twist}
implies $\sum_{i=2}^{\NumAgents}\MaxProb_i >
\sum_{i=2}^{\NumAgents}\QFun(\MaxVal_i)-\QFun(\MaxVal_{i-1}) =
\QFun(\MaxVal_\NumAgents)$. On the other hand
$\sum_{i=2}^{\NumAgents}\MaxProb_i \le 1$. Therefore
$\QFun(\MaxVal_\NumAgents) < 1$ which implies $\MaxVal_\NumAgents >
\SPPrice$ because $\QFun(\SPPrice)=1$ and, by \cref{res:mono},
$\QFun(\cdot)$ is decreasing.
\end{proof}

We conclude the section with the proof of the upper-bound of
\cref{main-theorem}.

\begin{proof}[Proof of upper-bound in \cref{main-theorem}]
It follows from program~\eqref{eq:ratio},
\cref{res:tri_dist,res:ratio2,res:tight}, and the rest of the
discussion in this \lcnamecref{sec:upperbound} that $\RatioII$ which
is computed by \eqref{eq:ratio3} is an upper bound on the ratio of the
ex ante relaxation to the expected revenue of the optimal anonymous
pricing.  Following \cref{lastprog:lemma}, $\RatioII$ is upper bounded
by the objective value of program \eqref{eq:ratio6}, i.e. $\Ratio''$.
As $\QFun(\cdot)$ and $\HFun(\cdot)$ are decreasing (\cref{res:mono}),
the optimal solution to program~\eqref{eq:ratio6} is given by
$\Ratio''=1+\HFun(\QFun^{-1}(1))$ which numerically evaluates to $e
\approx 2.718$.
\end{proof}

%% file: lowerbound.tex
\section{Lower-Bound Analysis}
\label{sec:lowerbound}

In this section we show the tightness of our approximation, i.e.\@ the lower-bound in  \cref{main-theorem}. As a result of  \cref{res:tri_dist}, it suffices to prove the following lemma.

 
\begin{lemma}
\label{lowerbound-lemma}
For any $\epsilon>0$ there exists a feasible assignment $(\NumAgents,\MaxVals,\MaxProbs)$ of the program \eqref{eq:tri_ratio} such that $\sum_{i=1}^\NumAgents \MaxVal_i\MaxProb_i\geq 1+\HFun(\QFun^{-1}(1))-\epsilon$.
\end{lemma}
\begin{proof}
 Pick $\delta>0$  s.t. $(1-\delta)^2\left (1+\HFun\left(\QFun^{-1}\left (\tfrac{1}{(1+\delta)^2}\right )+\delta\right)\right )\geq 1+\HFun(\QFun^{-1}(1))-\epsilon$. This is always possible as $\HFun$ and $\QFun$ are decreasing. Lets define $\lambda=\QFun^{-1}(1)$. The proof is done in two steps:\\
 
\noindent\emph{Step 1:} We find $\{\Val_i,\Prob_i\}_{i=2}^n$ such that
\begin{align*}
 &\sum_{i=2}^{n} \Prob_i\leq 1 ~~~,~~~k=2,\ldots,n:~~\sum_{i=2}^k \ln (1+\Val_i\Prob_i)= \HFun(\Val_k)\\
 &\sum_{i=2}^n \Val_i\Prob_i\geq \HFun\left(\QFun^{-1}\left (\tfrac{1}{(1+\delta)^2}\right )+\delta\right)
\end{align*}

In our construction for $\{\Val_i,\Prob_i\}_{i=2}^n$, we use two parameters $\Delta>0$ and $V_T\geq \lambda$ which we fix later in the proof. Let $\Val_1\triangleq\infty$,  and for $i\geq 2$ set $\Val_i=V_T-(i-2)\Delta$ and $\Prob_i=\frac{e^{\HFun(\Val_{i})-\HFun({\Val_{i-1}}) }-1}{\Val_i}$. Now, let $n=\max\{n_0\in \mathbb{N}: \sum_{i=2}^{n_0}\Prob_i\leq 1\}$. Obviously, $\sum_{i=2}^{n}\Prob_i\leq 1$. Moreover, for any $2\leq k \leq n$ we have $\sum_{i=2}^k \ln(1+\Val_i\Prob_i)=\sum_{i=2}^k (\HFun(\Val_i)-\HFun(\Val_{i-1}) )=\HFun(\Val_k)-\HFun(\Val_1)=\HFun(\Val_k)$. Now, pick $\delta'>0$ small enough such that for $x\in [0,\delta']$ we have $\frac{e^x-1}{x}\leq 1+\delta$. Moreover, let $\Delta$ to be small enough and $V_T$ to be large enough such that $\max\{\HFun(\lambda)-\HFun(\lambda+\Delta), \HFun(V_T),\Delta\}\leq \min\{\delta,\delta'\}$. First observe that due to Lemma~\ref{res:twist} $\Prob_{i}\geq \QFun(\Val_i)-\QFun(\Val_{i-1})$ which implies $2\geq \sum_{i=1}^n \Prob_i\geq\QFun(\Val_n)-\QFun(\Val_1)=\QFun(\Val_n)$. So,  all values $\Val_i$ are at least equal to $\lambda$. As $\HFun(.)$ is convex over $[1,\infty)$, we have $\HFun(\Val_i)-\HFun(\Val_{i-1})\leq \HFun(\lambda)-\HFun(\lambda+\Delta)\leq\delta'$. As a result we have
\begin{align}
\Prob_i&=\frac{e^{\HFun(\Val_{i})-\HFun({\Val_{i-1}}) }-1}{\Val_i}\leq (1+\delta) \frac{{\HFun(\Val_{i})-\HFun({\Val_{i-1}}) }}{\Val_i}= (1+\delta)\int_{\Val_i}^{\Val_{i-1}}\frac{\HFun'(\MaxValC) }{\Val_i}\dif\MaxValC.\nonumber\\
&= (1+\delta)\int_{\Val_i}^{\Val_{i-1}}-\frac{\Val}{\Val_i}\QFun'(\Val)\dif\MaxValC=(1+\delta) \left(\QFun(\Val_i)-\QFun(\Val_{i-1})+\int_{\Val_i}^{\Val_{i-1}}-\frac{v-\Val_i}{\Val_i}\QFun'(\Val)\dif\Val\right)\nonumber\\
&\leq (1+\delta)\left(\QFun(\Val_i)-\QFun(\Val_{i-1})+\Delta \int_{\Val_i}^{\Val_{i-1}}-\QFun'(\Val)\dif\Val\right)\leq (1+\delta)^2(\QFun(\Val_i)-\QFun(\Val_{i-1}))\label{lemma1-lower}
\end{align}
Based on the definition of $n$ (number of distributions in our instance),  we have $1< \sum_{i=2}^{n+1}\Prob_i$. By (\ref{lemma1-lower}), we have $\sum_{i=2}^{n+1}\Prob_i\leq (1+\delta)^2 \sum_{i=2}^{n+1}((\QFun(\Val_i)-\QFun(\Val_{i-1}))=(1+\delta)^2\QFun(\Val_{n+1})$. Lets define $\lambda'\triangleq \QFun^{-1}\left (\tfrac{1}{(1+\delta)^2}\right )$. We conclude that $\lambda'\geq \Val_{n+1}$. Hence,  $\Val_n\leq \lambda'+\Delta\leq \lambda'+\delta$. Moreover, using Lemma~\ref{res:tri_dist} we have
\begin{equation}
\sum_{i=2}^n \Val_i\Prob_i\geq \sum_{i=2}^n (\HFun(\Val_i)-\HFun(\Val_{i-1}))=\HFun(\Val_n)\geq \HFun(\lambda'+\delta)=\HFun\left(\QFun^{-1}\left(\tfrac{1}{(1+\delta)^2}\right)+\delta\right)
\end{equation}
where the last inequality is true because $\Val_n\leq \lambda'+\delta$ and $\HFun$ is decreasing over $[1,\infty)$.\\

\noindent\emph{Step 2:} Given $\{\Val_i,\Prob_i\}_{i=2}^n$, we find an instance $\{\MaxVal_i,\MaxProb_i\}_{i=1}^n$ such that is feasible for program \eqref{eq:tri_ratio} and $\sum_{i=1}^n \MaxVal_i\MaxProb_i\geq 1+\HFun(\QFun^{-1}(1))-\epsilon$. To do so, set
 $\Prob_1=\delta, \Val_1=\frac{1}{\delta}-1$. Now, for each $ i,k\in[2:n]$ find $\gamma_{i,k}$ such that
\begin{equation}
\label{eq1-lower}
1+\frac{\Val_i\Prob_i(1-\gamma_{i,k})}{\Val_k}=(1+\Val_i\Prob_i)^{\frac{1}{\Val_k}}
\end{equation}
and then let $\MaxProb_i=(1-\delta)(1-\underset{k\in [2:n]}{\max}\gamma_{i,k})\Prob_i$ and $\MaxVal_i=\Val_i$, for $i\in [2:n]$.  Now we claim $\{{\MaxVal_i},\MaxProb_i\}_{i=1}^n$ is a feasible assignment for the program \eqref{eq:tri_ratio} . We have
\begin{equation}
\sum_{i=1}^{n}\MaxProb_i=\delta+(1-\delta)\sum_{i=2}^n (1-\underset{k\in [2:n]}{\max}\gamma_{i,k})\Prob_i\leq \delta+(1-\delta)\sum_{i=2}^{n}\Prob_i\leq 1.
\end{equation}
as $\sum_{i=2}^{n}\Prob_i\leq 1$. Moreover, for $k\in[2:n]$ we have
\begin{align*}
\sum_{i=1}^k\ln\left (1+\frac{\MaxVal_i\MaxProb_i}{\MaxVal_k(1-\MaxProb_i)}\right )&\leq \ln\left (1+\frac{\MaxVal_1\MaxProb_1}{\MaxVal_k(1-\MaxProb_1)}\right )+\sum_{i=2}^k\ln\left (1+\frac{\Val_i\Prob_i(1-\gamma_{i,k})}{\Val_k}\right )\nonumber\\
&=\ln\left (\frac{\Val_k+1}{\Val_k}\right )+\frac{1}{\Val_k}\sum_{i=2}^k\ln(1+\Val_i\Prob_i)\\
&\leq \ln\left (\frac{\Val_k+1}{\Val_k}\right )+\ln\left (\frac{\Val_k^2}{\Val_k^2-1}\right )\\
&=\ln\left (\frac{\MaxVal_k}{\MaxVal_k-1}\right )
\end{align*} 
By taking exponents from both sides and rearranging the terms it is not hard to see $(n,\MaxVals,\MaxProbs)$ is a feasible assignment of program \eqref{eq:tri_ratio}. Additionally,  for a fixed $V_T$ all of the $\Val_i$'s are bounded, i.e. $1\leq \Val_i\leq V_T$. So, as $\Delta$ goes to zero we have $\Val_i\Prob_i\rightarrow 0 $ as $\Prob_i\rightarrow 0$, and the left hand side of (\ref{eq1-lower}) converges to its right hand side. As a result, for small enough $\Delta$, we can guarantee  $\gamma_{i,k}\leq \delta$ for all $i,k$, and hence $\MaxProb_i\geq (1-\delta)^2 \Prob_i$. So 
\begin{align*}
\sum_{i=1}^n \MaxVal_i\MaxProb_i&\geq (1-\delta)+(1-\delta)^2\left (\sum_{i=2}^n \Val_i\Prob_i\right )\\
&\geq (1-\delta)^2 \left (1+\sum_{i=2}^n \Val_i\Prob_i\right )\\
&\geq (1-\delta^2)\left(1+\HFun\left (\QFun^{-1}\left (\tfrac{1}{(1+\delta)^2}\right )+\delta\right )\right)
\end{align*}
which implies $\sum_{i=1}^n \MaxVal_i\MaxProb_i\geq 1+\HFun(\QFun^{-1}(1))-\epsilon$, as desired.
\end{proof}

%% file: irregular.tex
\section{Irregular inapproximability results}
\label{s:irregular}

In this section we show that anonymous pricing and anonymous reserves
are a tight $n$ approximation to the optimal auction and ex ante
relaxation.  Specifically, we show a lower bound on the approximation
factor of anonymous reserves to the optimal auction and an upper bound
on the approximation factor of anonymous pricing to the ex ante
relaxation.  The ordering of these mechanisms by revenue then implies
that all bounds are optimal and tight.

\begin{proposition}
\label{t:irregular-monop-reserve-linear-lb}
For $n$-agent, independent, non-identical, and irregular distributions
the second-price auction with anonymous reserves is at best an $n$
approximation to the optimal single-item auction.
\end{proposition}

\begin{proof}
Consider the following value distribution
$$
v_i  =
\begin{cases}
		h^i  & \text{with probabiliy $h^{-i}$,}\\
		0 & \text{otherwise.}
\end{cases}
$$
On this distribution the ex ante relaxation has revenue
$\sum_{i=0}^{n} h^i h^{-i} = n$ (and the optimal auction is no
better).  On the other hand, anonymous reserve and anonymous pricing
of $h^i$ for any $i \in \{1,\ldots,n\}$ gives revenue at least one.
We will show that in the limit as $h$ approaches infinity; these
bounds are tight.

We first argue that in the limit of $h$ the optimal auction revenue is
$n$ (the same as the ex ante relaxation).
%
%
%
Consider the expected revenue of the following sequential posted
pricing mechanism, which gives a lower-bound on the optimal
revenue.\footnote{In fact, this sequential posted pricing mechanism is
  the optimal auction, but its optimality is unnecessary for the proof
  so we omit the details.}  In decreasing order of price and until the
first agent accepts her offered price, offer each agent $i$ price
$h^i$.  This mechanism's revenue can be calculated as:
$$
h^n\cdot \tfrac{1}{h^n}+\sum_{i=2}^n{  h^{n-i+1}}\cdot \tfrac{1}{h^{n-i+1}} \prod_{j=1}^{i-1}(1-\tfrac{1}{h^{n-j+1}}) =1+\sum_{i=2}^n{\prod_{j=1}^{i-1}(1-\tfrac{1}{h^{n-j+1}})}
$$
which converges to $n$ as $h$ goes to infinity. 

We now prove that the expected revenue of the second-price auction with any
anonymous reserve in $\{h^i\}_{i=1}^n$ is at most one in the limit.
(Any other reserve is only worse.)  In the second-price auction with
reserve, the winner pays the maximum of the highest agent value below
her value and the reserve.  An upper bound of this revenue is the sum
over all agents with values at least the reserve. So, for reserve $h^i$ the contribution of $j \geq i$ to this upper bound
is at most:
\begin{equation}
\label{eq:anon-reserve-bound}
h^{-j} \,(j-i + h^i).
\end{equation}
The first term, above, is the probability that agent $j$ has high
value $h^j$. Conditioned on her having the high value $h^j$, the
second term bounds the agent's payment, the expected maximum of the
highest lower-valued agent and the reserve $h^i$.  It is at most $j-i
+ h^i$ as each agent between $i$ and $j$ has expected value one and
the expectation of their maximum is at most the sum of their
expectations.  It follows from equation~\eqref{eq:anon-reserve-bound} that in the limit with $h$, the contribution
from agent $i$ to this bound is one and the contribution from agent $j
\neq i$ is zero.  Thus, the expected revenue of the second-price
auction with reserve $h^i$ is at most one in the limit.
\end{proof}

\begin{proposition} 
\label{t:irregular-monop-pricing-linear-ub}
For independent, non-identical, irregular $n$-agent single-item
environments, anonymous pricing is at worst an $n$ approximation
to the ex ante relaxation.
\end{proposition}

\begin{proof}
Define $\{(\MaxVal_i,\MaxProb_i)\}_{i=1}^n$ as in
equation~\eqref{eq:tri_bench_rev} in \cref{sec:prelim} where the ex ante relaxation posts
price $\MaxVal_i$ to agent $i$ which is accepted with probability
$\MaxProb_i$ and has total revenue $\sum\nolimits_{i=1}^\NumAgents
\MaxVal_i\,\MaxProb_i$.  For any $i$ the anonymous pricing that posts
price $\MaxVal_i$ obtains at least revenue $\MaxVal_i\,\MaxProb_i$.
Thus, picking a uniformly random price from $\{\MaxVal_i\}_{i=1}^{n}$ gives
an $n$ approximation to the ex ante relaxation revenue
$\sum\nolimits_{i=1}^\NumAgents \MaxVal_i\,\MaxProb_i$.  The optimal
anonymous price is no worse.
\end{proof}



%


%% file: simulation.tex
\section{Simulation results}
\label{appendix:sim}
In this section, we briefly discuss simulation results for the
worst-case instance derived in section \secref{sec:upperbound}.  From
these simulations we will see how fast, as a function of the number
$n$ of agents, the worst-case ratio of the ex ante relaxation to the
expected revenue of optimal anonymous pricing converges to $e$.
Moreover, for these worst-case instances we will be able to evaluate
the approximation of the optimal auction by anonymous reserves and
pricing.

Our worst-case instances are given by a continuum of agents (as given
by $\HFun(\cdot)$ and $\QFun(\cdot)$) which we discritize by
evaluating $\QFun(\cdot)$ on an arithmetic progression, denoted
$\{\MaxProb_i\}$.  Given the fast convergence that our simulation
exhibits, there is little loss in this discritization.

According to the price revenue constraint in \eqref{eq:tri_ratio}, we
know that the revenue by posting price $\MaxVal_i$ for every $i$
should be equal to the optimal posting price revenue. Thus, when we
have $\{\MaxProb_i\}$, it is simple to get the corresponding
$\MaxVal_i$ by binary search and calculating the revenue.

After generating the instances, we also calculate the ratio of the
revenue of the optimal mechanism to the anonymous pricing revenue, and
the ratio of the revenue of the optimal mechanism to the revenue of
the second price with anonymous reserve mechanism for these
instances. We use sampling algorithm to calculate the revenue of the
second price with anonymous reserve mechanism, while the calculation
of the revenue of the optimal mechanism is exact. We report the
results of our simulation in \cref{tab:sim,fig:sim} for various
numbers $n$ of agents.

\begin{figure}[ht]
\centering
\begin{tabular}{|c|c|c|c|c|c|c|c|c|}
\hline
n                          &2& 5&10 & 50&100& 500    & 1000   & 5000    \\
\hline
$ \BenchRev/\OptPriceRev $   &2.000& 2.507&  2.622& 2.701& 2.710&
2.717 &
 2.718&
2.718  \\\hline
$ \textsc{OptRev}/\OptPriceRev $      &2.000& 2.138 & 2.187 & 2.223&2.227 & 2.231&2.231& 2.232\\
\hline
$ \textsc{OptRev}/\textsc{OptReserveRev} $      &2.000& 1.794   & 1.731     &   1.682 &1.676&
1.665&1.659&1.607 \\
\hline
\end{tabular}
\caption{The ratios of the revenues of various auctions and
  benchmarks. Here $\BenchRev$ and $\OptPriceRev$ are the ex ante
  relaxation and optimal anonymous pricing revenues (as previously
  defined). $\textsc{OptRev}$ is the revenue of the optimal auction of
  \citet{mye-81}. $\textsc{OptReserveRev}$ is the revenue obtained by
  the second-price auction with an optimally chosen reserve price.}
\label{tab:sim}
\end{figure}


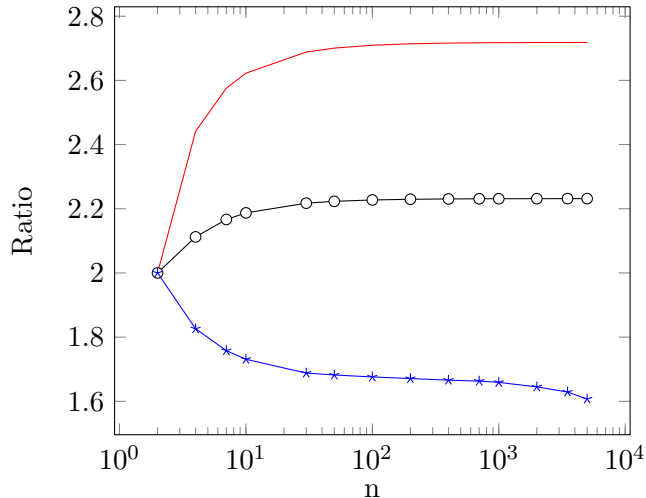
\begin{figure}[ht]
\centering
\begin{tikzpicture}

\begin{semilogxaxis}[
  xlabel=n,
  ylabel=Ratio,
  cycle list={%
{red},
{black,mark=*,fill=white},
{blue,mark=star}}
]
\addplot table [y=B, x=A]{tmp.txt};
\addplot table [y=C, x=A]{tmp.txt};
\addplot table [y=D, x=A]{tmp.txt};
\end{semilogxaxis}
\end{tikzpicture}

\caption{The ratios of the revenues of various auctions and benchmarks.
The red sold line represents the ratio of ex-ante benchmark to anonymous pricing revenue. The blue line with star represents the ratio of the optimal revenue to anonymous pricing revenue. The black line with circle represents the ratio of the optimal revenue to 
the revenue of the second price auction with reserve.
}
\label{fig:sim}
\end{figure}
We conclude this section by highlighting the following theorem, derived from the table in \cref{tab:sim}.
\begin{theorem}
There exists an instance for which anonymous pricing is a 2.232 approximation to the optimal auction.
\end{theorem}

%% file: upperbound-appendix.tex
\section{Missing Proofs from \secref{sec:upperbound}}

\phantomsection
\label{proof:res:ratio2}
\begin{lemma*}[\ref{res:ratio2}]
The value of program \eqref{eq:tri_ratio2}, denoted by $\RatioII$, is an upper bound on the value of program
\eqref{eq:tri_ratio} which is $\Ratio$.
\end{lemma*}
\begin{proof}
Let $(\MaxVals',\MaxProbs')$ be any feasible assignment for program~\eqref{eq:tri_ratio} for which $\sum_{i=1}^\NumAgents
\MaxVal'_i\MaxProb'_i > 1$. We construct a corresponding assignment $(\MaxVals,\MaxProbs)$ which is feasible for
program~\eqref{eq:tri_ratio2} and yields the same objective value, that is $1+\sum_{i=2}^\NumAgents
\MaxVal_i\MaxProb_i=\sum_{i=1}^\NumAgents \MaxVal'_i\MaxProb'_i$ which then implies that $\BenchRatioI \le
\BenchRatioII$. Without loss of generality assume $\MaxVal'_1\ge \ldots \ge \MaxVal'_\NumAgents$. Let $j$ be the
smallest index for which $\sum_{i=1}^j \MaxVal'_i\MaxProb'_i > 1$. Observe that $2 \le j \le \NumAgents$ because
$\MaxVal'_i\MaxProb'_i \le 1$ for all $i$. Let $\delta=1-\sum_{i=1}^{j-1} \MaxVal'_i\MaxProb'_i$. Observe that $0 \le
\delta < \MaxVal'_j\MaxProb'_j$. We construct a new optimal assignment $(\MaxVals,\MaxProbs)$ by setting for each $i \in
\RangeMN{2}{\NumAgents}$:
\begin{align*}
        \MaxVal_i &=
            \begin{cases}
            \MaxVal'_{j+i-2}      & 2 \le i \le \NumAgents-j+2 \\
            1                   & \NumAgents-j+3 \le i \le \NumAgents
            \end{cases} &
        \MaxProb_i &=
            \begin{cases}
            \MaxProb'_{j}-\frac{\delta}{\MaxVal'_j}  & i=2 \\
            \MaxProb'_{j+i-2}                       & 3 \le i \le \NumAgents-j+2 \\
            0                                   & \NumAgents-j+3 \le i \le \NumAgents
            \end{cases}.
    \end{align*}
By the above construction it is easy to see that $1+\sum_{i=2}^\NumAgents \MaxVal_i\MaxProb_i=\sum_{i=1}^\NumAgents
\MaxVal'_i\MaxProb'_i$. So we only need to show that $(\MaxVals,\MaxProbs)$ is indeed a feasible assignment. 
Observe that $\sum_{i=2}^\NumAgents \MaxProb_i \le \sum_{i=1}^\NumAgents \MaxProb'_i$, so the second constraint holds. So we
only need to show that $(\MaxVals,\MaxProbs)$ satisfies the constraint (\ref{firstconstraintP4}).

By rearranging \eqref{eq:tri_ratio:price_cons} we get
\begin{align}
  \prod_{i:\MaxVal'_i \ge \PPrice}\left(1+\frac{\MaxVal'_i \MaxProb'_i}{\PPrice\cdot(1-\MaxProb'_i)}\right) &\le
      \left(\frac{\PPrice}{\PPrice-1}\right) & & \forall \PPrice > 0. \notag \\
\intertext{We then relax the previous inequality by dropping the term $(1-\MaxProb'_i)$ from the denominator of the left hand side and
    take the logarithm of both sides to get}
  \sum_{i:\MaxVal'_i \ge \PPrice} \ln\left(1+\frac{\MaxVal'_i \MaxProb'_i}{\PPrice}\right) &\le
      \ln\left(\frac{\PPrice}{\PPrice-1}\right) & & \forall \PPrice > 0. \label{eq:barrier1}
\end{align}
On the other hand, by invoking~\cref{lem:logsum} we can argue
that~\footnote{We invoke \cref{lem:logsum} by setting
$b=\frac{1}{\MaxVal_k},a=\frac{\MaxVal_2\MaxProb_2}{\MaxVal_k},m=j$ and $z_i=\frac{\MaxVal'_i\MaxProb'_i}{\MaxVal_k}$.}
\begin{align}
    \ln\left(1+\frac{1}{\MaxVal_k}\right)+\sum_{i=2}^k \ln\left(1+\frac{\MaxVal_i \MaxProb_i}{\MaxVal_k}\right) & \le
         \sum_{i=1}^{\min(k+j-2,\NumAgents)} \ln\left(1+\frac{\MaxVal'_i \MaxProb'_i}{\MaxVal_k}\right) 
         & & \forall k \in \RangeMN{2}{\NumAgents}. \label{eq:revamp}
\end{align}
Observe that for any given $k$ the the right hand side of \eqref{eq:revamp} is equal or less than the the left hand side of
\eqref{eq:barrier1} for $\PPrice=\MaxVal_k$ which implies
\begin{align*}
    \ln\left(1+\frac{1}{\MaxVal_k}\right)+\sum_{i=2}^k \ln\left(1+\frac{\MaxVal_i \MaxProb_i}{\MaxVal_k}\right) &
    \le  \ln\left(\frac{\MaxVal_k}{\MaxVal_k-1}\right) & & \forall k \in \RangeMN{2}{\NumAgents}.
\end{align*}
We can then further relax the above inequality by replacing the terms $\ln(1+\frac{\MaxVal_i \MaxProb_i}{\MaxVal_k})$
with $\frac{1}{\MaxVal_k}\ln(1+\MaxVal_i \MaxProb_i)$ and rearranging the terms to get the constraint (\ref{firstconstraintP4})
which implies that $(\MaxVals,\MaxProbs)$ is feasible with respect to that constraint as well.
\end{proof}

\phantomsection
\label{proof:res:mono}
\begin{lemma*}[\ref{res:mono}]
The functions $\HFun(\PPrice)$, $\QFun(\PPrice)$, and $\HFun(\PPrice)-\QFun(\PPrice)$ are all decreasing in $\PPrice$,
for $\PPrice > 1$.
\end{lemma*}
\begin{proof}
To prove $\HFun(\PPrice)$ is decreasing, we show its derivative is negative:
\begin{align*}
    \HFun'(\PPrice)
            &= \ln\left(1+\frac{1}{\PPrice^2-1}\right)-\frac{2}{\PPrice^2-1}
            < \frac{1}{\PPrice^2-1}-\frac{2}{\PPrice^2-1}
            <0.
\end{align*}
The first inequality uses the fact that $\ln(1+x)\le x$. Similarly $\QFun(\PPrice)$ is also decreasing because
$ \QFun'(\PPrice) = \frac{1}{\PPrice} \HFun'(\PPrice) < 0$. Finally
$\HFun(\PPrice)-\QFun(\PPrice)$ is also decreasing because
\begin{align*}
  \left(\HFun(\PPrice)-\QFun(\PPrice)\right)' &= \HFun'(\PPrice)\left(1-\frac{1}{\PPrice}\right) <
    0.
\end{align*}

\end{proof}

\phantomsection
\label{proof:res:hbound}
\begin{lemma*}[\ref{res:hbound}]
$
\HFun(\PPrice)-\HFun(\PPrice') < \ln(\frac{\PPrice}{\PPrice-1})-\ln(\frac{\PPrice'}{\PPrice'-1})$ for any $\PPrice'>
\PPrice > 1$.
\end{lemma*}
\begin{proof}
%
Define $\GFun(\PPrice)=\HFun(\PPrice)-\ln(\frac{\PPrice}{\PPrice-1})$. Observe that proving the inequality in the
statement of the lemma is equivalent to proving $\GFun(\PPrice)< \GFun(\PPrice')$ which we do by showing
$\GFun(\PPrice)$ has positive derivative and is therefore increasing.
\begin{align*}
   \GFun'(\PPrice) &= \ln\left(1+\frac{1}{\PPrice^2-1}\right)-\frac{1}{\PPrice^2+\PPrice}.
\end{align*}
We will show that $\GFun'(\PPrice)$ is decreasing which then implies $\GFun'(\PPrice) > 0$
because $\lim_{\PPrice\to\infty}\GFun'(\PPrice)=0$. Therefore we only need to show that
$\GFun''(\PPrice)<0$.
\begin{align*}
    \GFun''(\PPrice) &= \frac{-(3\PPrice+1)}{(\PPrice-1)\PPrice^2(\PPrice+1)^2} < 0.
\end{align*}
\end{proof}

\begin{lemma}
\label{lem:logsum}%
Consider any $a,b, z_1,\ldots,z_m \ge 0$ such that $a+b=\sum_{i=1}^m z_i$ and  $a \le z_m \le b$. Then
    \begin{align*}
        \ln(1+b)+\ln(1+a) \le \sum_{i=1}^m \ln(1+z_i).
    \end{align*}
\end{lemma}

\begin{proof}
We can re-write the equation as
\begin{equation}\label{eq:lemd1}
\ln \left( (1+a)(1+b) \right) \leq \ln \prod_{i=1}^m (1+z_i)
\end{equation}

Observe that 
\begin{eqnarray*}
\prod_{i=1}^m (1+z_i) &\geq & 1+\sum_{i=1}^m z_i + z_m (\sum_{i=1}^{m-1}z_i)
\\ &=& (1+\sum_{i=1}^{m-1}z_i) (1+z_m)
\\
&\geq & (1+a)(1+b),
\end{eqnarray*}

where the first inequality follows by eliminating some terms from the expansion of $\prod_{i=1}^m (1+z_i)$, and the second inequality from the assumption that $(1+a)\leq (1+z_m)$ and  $(1+b)\geq (1+\sum_{i=1}^{m-1}z_i)$.  
\end{proof}

%% file: main-anpp.bbl
\begin{thebibliography}{}

\bibitem[Alaei, 2011]{ala-11}
Alaei, S. (2011).
\newblock Bayesian combinatorial auctions: Expanding single buyer mechanisms to
  many buyers.
\newblock In {\em Proceedings of the 52nd Annual Symposium on Foundations of
  Computer Science}, pages 512--521. IEEE.

\bibitem[Alaei et~al., 2013]{AFHH-13}
Alaei, S., Fu, H., Haghpanah, N., and Hartline, J. (2013).
\newblock The simple economics of approximately optimal auctions.
\newblock In {\em Foundations of Computer Science (FOCS), 2013 IEEE 54th Annual
  Symposium on}, pages 628--637. IEEE.

\bibitem[Bulow and Klemperer, 1996]{BK-96}
Bulow, J. and Klemperer, P. (1996).
\newblock Auctions versus negotiations.
\newblock {\em American Economic Review}, 86:180--194.

\bibitem[Cai and Daskalakis, 2011]{CD-11}
Cai, Y. and Daskalakis, C. (2011).
\newblock Extreme-value theorems for optimal multidimensional pricing.
\newblock In {\em Proceedings of the 52nd Annual Symposium on Foundations of
  Computer Science}, pages 522--531. IEEE.

\bibitem[Chawla et~al., 2007]{CHK-07}
Chawla, S., Hartline, J., and Kleinberg, R. (2007).
\newblock Algorithmic pricing via virtual valuations.
\newblock In {\em Proc. 8th ACM Conf. on Electronic Commerce}.

\bibitem[Chawla et~al., 2010a]{CHMS-10}
Chawla, S., Hartline, J., Malec, D., and Sivan, B. (2010a).
\newblock Sequential posted pricing and multi-parameter mechanism design.
\newblock In {\em Proc. 41st ACM Symp. on Theory of Computing}.

\bibitem[Chawla et~al., 2010b]{CMS-10}
Chawla, S., Malec, D., and Sivan, B. (2010b).
\newblock The power of randomness in bayesian optimal mechanism design.
\newblock In {\em Proc. 11th ACM Conf. on Electronic Commerce}, pages 149--158.

\bibitem[Chen et~al., 2014]{CGL-14}
Chen, N., Gravin, N., and Lu, P. (2014).
\newblock Optimal competitive auctions.
\newblock In {\em Proceedings of the 46th Annual ACM Symposium on Theory of
  Computing}, pages 253--262. ACM.

\bibitem[Goldberg et~al., 2006]{GHKSW-06}
Goldberg, A., Hartline, J., Karlin, A., Saks, M., and Wright, A. (2006).
\newblock Competitive auctions.
\newblock {\em Games and Economic Behavior}, 55:242--269.

\bibitem[Haghpanah and Hartline, 2014]{HH-14}
Haghpanah, N. and Hartline, J. (2014).
\newblock Reverse mechanism design.
\newblock {\em arXiv preprint arXiv:1404.1341}.

\bibitem[Hartline, 2013]{har-13}
Hartline, J. (2013).
\newblock Bayesian mechanism design.
\newblock {\em Foundations and Trends{\textregistered} in Theoretical Computer
  Science}, 8(3):143--263.

\bibitem[Hartline and Roughgarden, 2009]{HR-09}
Hartline, J. and Roughgarden, T. (2009).
\newblock Simple versus optimal mechanisms.
\newblock In {\em Proc. 10th ACM Conf. on Electronic Commerce}.

\bibitem[Myerson, 1981]{mye-81}
Myerson, R. (1981).
\newblock Optimal auction design.
\newblock {\em Mathematics of Operations Research}, 6:58--73.

\bibitem[Yan, 2011]{yan-11}
Yan, Q. (2011).
\newblock Mechanism design via correlation gap.
\newblock In {\em Proceedings of the Twenty-Second Annual ACM-SIAM Symposium on
  Discrete Algorithms}, pages 710--719. SIAM.

\end{thebibliography}
